\documentclass[draft]{article}
\usepackage{amsmath,amsfonts,latexsym, amssymb}

\usepackage{color}

\newcommand{\Tr}{\mbox{{\rm Tr}\,}}

\newcommand{\pt}{\partial}
\newcommand{\rd}{{\rm d}}

\newcommand{\bs}{ \begin{split} }
\newcommand{\es}{ \end{split} }

\newcommand{\bE}{{\bf{E}}}

\newcommand{\al}{\alpha}

\newcommand{\be}{\begin{equation}}
\newcommand{\ee}{\end{equation}}

\newcommand{\om}{{\omega}}

\newcommand{\T}{{\mathbb T}}

\newcommand{\Z}{{\mathbb Z}}
\newcommand{\E}{{\bf{E}}}
\newcommand{\R}{{\mathbb R}}
\newcommand{\N}{{\mathbb N}}

\newcommand{\wt}{\widetilde}

\newcommand{\FF}{\mathcal{F}}

\newcommand{\AR}{\mathcal{E}}
\newcommand{\HH}{\mathcal{H}}

\newtheorem{theorem}{Theorem}

\newtheorem{lemma}[theorem]{Lemma}
\newtheorem{proposition}[theorem]{Proposition}

\newcommand{\qed}{\hfill\fbox{}\par\vspace{0.3mm}}
\newenvironment{proof}{{\bf Proof.}} {\hfill\qed}

\numberwithin{equation}{section}
\numberwithin{theorem}{section}
\numberwithin{definition}{section}

\title{Anderson localization for random magnetic Laplacian on $\Z^2$}
\author{L\'aszl\'o Erd\H os\thanks{Partially supported by SFB-TR12 of
the German Science Foundation},\; David Hasler
 \\
\\
Institute of Mathematics, University of Munich, \\
Theresienstr. 39, D-80333 Munich, Germany \\
\text{lerdos@math.lmu.de, hasler@math.lmu.de} }

\date{Jan 11, 2011}

\begin{document}

\maketitle

\begin{abstract}
We consider a two dimensional magnetic Schr\"odinger operator on  a square lattice
with a spatially stationary random magnetic field. We prove 
 Anderson localization near the spectral edges.
We use a new approach to establish a Wegner estimate that does not rely
on the monotonicity of the energy on the random  parameters.

\end{abstract}


\medskip

{\it Running title:} Localization for magnetic Laplacian

\medskip



\section{Introduction}

We consider a spinless quantum particle hopping on the two dimensional
lattice $\Z^2$ and subject to a random magnetic field.
 This model is the magnetic
analogue of the standard Anderson model  with a random on-site potential
but here the magnetic field carries the randomness in the system.
The fluxes through the plaquets  are i.i.d. random variables. 
For simplicity we assume that there is no external potential.

The main result of this paper is a Wegner estimate for
the averaged density of states for the Hamiltonian
restricted to a finite box $\Lambda$.
More precisely,  we show that the expected number of eigenvalues
in a small spectral interval  of length $\eta$ is bounded
by $C\eta |\Lambda|^4$. This estimate exhibits the
optimal (first) power of $\eta$, but its volume dependence
is not optimal. Nevertheless, it can be used to prove spectral and
dynamical localization via the standard multiscale argument.

A similar result  has been obtained earlier 
by Klopp {\it et al.} in \cite{KNNN}, but under a
quite restrictive condition, namely that in a fixed domino tiling of $\Z^2$
the flux on each domino is deterministically zero.
This condition was essential for the method of \cite{KNNN} to work
since it ensured that the magnetic field was generated by
a stationary vector potential that could be expressed
in terms of independent gauges on each domino.
The variation of the local flux thus influenced
the quantum state only on a few sites.  Diagonalization of a finite matrix then showed that
near the spectral edges
the energy is a strictly monotone function 
of the flux. This monotonicity observation provided the
key input for the Wegner estimate in \cite{KNNN}.

If the zero flux condition is removed and 
we consider independent fluxes on each plaquets,
then apparently there is no direct 
monotone relation between the eigenvalues  and the fluxes.
Prior to our recent work \cite{EH2},
some form of monotonicity has always been used in 
the proofs of the Wegner type estimates in context
of random Schr\"odinger operators (discrete or continuum).
In case of random external potential with a definite sign,
the monotonicity of an eigenvalue as a function
of the random coupling constants is  a direct
consequence of first order perturbation theory.
For sign indefinite potentials \cite{HK}, for
the random displacement model \cite{KLNS}
 and for random vector potentials \cite{HK, GHK, U}
the  monotonicity could still be extracted
using some special structure of these models
but similar ideas do not seem to apply for
random magnetic fields. Note that i.i.d.
(or stationary) random vector potentials 
and i.i.d. random magnetic fields represent 
different physical models; typically a stationary random
magnetic field cannot be generated by a 
stationary random vector potential.
We refer to
Section 4 of \cite{EH2}  for an overview and
further references.

In a recent paper \cite{EH2} we have developed a new
method to prove a Wegner estimate without any monotonicity mechanism.
We considered the corresponding continuous model, i.e.,
the Schr\"odinger operator in $\R^2$ with a  stationary random magnetic field.
We proved a Wegner estimate for all energies 
and Anderson localization at the bottom of the spectrum.
The magnetic field was required to have a positive lower bound
to ensure that the current did not vanish.
Furthermore, the random field had to contain 
modes on arbitrary small scales to control
the large momentum regime of the current.  

In the current work we extend this approach to the discrete case,
where many  technical complications of the continuum model
are absent and  we can thus
consider very general  random magnetic fields.
Apart from independence,
the only essential condition on the fluxes is that 
they should be separated away from 0 and $\pi$,
i.e. away from the ``minimal'' and ``maximal'' fluxes on
each plaquets. This condition is analogous
to the positive lower bound on the magnetic field
in the continuous model. 
The proof presented here is very simple and it highlights
the essence of our new approach introduced in \cite{EH2}.

\section{Model and Statements of Results}

For any subset $\Lambda\subset \Z^2$ we introduce
the Hilbert space $\ell^2(\Lambda)$, with 
inner product 
$$(\varphi_1, \varphi_2 ) = \sum_{x \in \Z^2} \overline{\varphi_1(x)} \varphi_2(x) .$$
The discrete magnetic Schr\"odinger operator will be defined on $\HH: = \ell^2(\Z^2)$.
For the precise definitions we will follow the notations 
 in   \cite{KNNN}, see also  \cite{Na1}. 

Let $\AR$ be the set of directed edges ({\it arrows}) in $\Z^2$, i.e., 
$$
\AR =\{ (x,y) : x,y \in \Z^2 , |x-y|=1 \} .
$$
For $a =(x,y) \in \AR$, we write $\overline{a} = (y,x)$. For $a \in \AR$ we will
denote by  $a_i$ and $a_t$ the first respectively second entry of $a$, i.e.,
 $a = (a_i,a_t)$.
 Let 
$\FF$ be the set of unit squares ({\it plaquets})  in $\Z^2$, i.e.,
$$
\FF =\left\{ \{ x_1,x_1+1 \} \times \{ x_2,x_2+1 \}  \ : \ (x_1,x_2) \in \Z^2 \right\}  .
$$
We label the elements of $\FF$ by the lattice point in the lower left corner of
the unit square, i.e., for $x = (x_1,x_2) \in \Z^2$ we define
 $f_x := \{x_1,x_1+1 \} \times \{ x_2,x_2+1 \} \in \mathcal{F}$. For any $f_x\in \FF$
we define the oriented boundary 
$$\partial f_x  = \left\{ (x, x + {\bf e}_1 ), 
(x + {\bf e}_1, x + {\bf e}_1 + {\bf e}_2 ), 
(x + {\bf e}_1 + {\bf e}_2, x + {\bf e}_2 ), 
(x + {\bf e}_2 , x ) \right\} \subset \AR ,
$$
where  ${\bf e}_1 = (1,0)$, and ${\bf e}_2 = (0,1)$.

Let $\mathbb{T} := \R / ( 2 \pi \Z)$.
 A  function $A: \AR \to \mathbb{T}$ will be called
{\it vector potential} or {\it gauge} if it satisfies
\be \label{defofG}
A(a) = - A(\overline{a}) ,\quad \forall a \in \AR .
\ee
Let 
$$
\mathcal{G} :=\{ A \in \mathbb{T}^{\AR} : \mbox{\eqref{defofG}  holds} \}
$$
be the space of vector potentials.
For $(x,y) \in \AR$ we will also use the notation  $A(x,y) := A((x,y))$.  
For a given vector potential $A \in \mathcal{G}$, the define the differential ``curl''  by 
$$
\rd A(f) :=  \sum_{a \in \partial f} A(a) , \quad f \in \FF ,
$$
which is a function on $\mathcal{F}$ with values in $\mathbb{T}$. 
We will call  $\rd A$ the {\it magnetic field} generated by  $A$.

The discrete magnetic Schr\"odinger operator $H$ 
on $\Z^2$ with a vector potential $A \in \mathcal{G}$ is defined 
by 
$$
(H(A) \psi)(x) = \sum_{y \in \Z^2 : |x-y|=1} \left( \psi(x) - e^{i A(x,y)} \psi(y) \right), 
\qquad \psi\in \ell^2(\Z^2).
$$
In view of the definition of $\mathcal{G}$ the operator $H(A)$ is self-adjoint  and 
 its quadratic form is given by
\begin{align*}
( \psi , H \psi ) =  \frac{1}{2} 
\sum_{a \in \AR} \left| \psi(a_i)  - e^{i A(a)} \psi(a_t) \right|^2  .
\end{align*}
We have the bound $0 \leq H(A) \leq 8$, \cite{Na1}, and it 
 is well known that $\sigma(H(0)) = [0,8]$.

\bigskip
\noindent
{\bf Remark.}
A function  $\lambda : \Z^2 \to \mathbb{T} $  defines a gauge transformation.
$U_\lambda \psi(x) := e^{i \lambda(x)} \psi(x)$, in the sense that 
 $H(A_\lambda) =  U_\lambda H(A) U_\lambda^*$ with  $A_\lambda(a) := A(a) + \rd \lambda(a)$ and  
$\rd \lambda(a) := \lambda(a_i)- \lambda(a_t)$.
It is well-known that if $A, \widetilde{A} \in \mathcal{G}$ generate the same magnetic field,
$\rd A=\rd \wt A$, then 
$H(A)$ and  $H(\widetilde{A})$ are unitarily equivalent by means of a gauge transformation.
 Thus spectral properties of $H(A)$ 
only depend on the underlying magnetic field.

\bigskip

Next we will introduce  the random magnetic field. 
Let  $B_\omega = (\omega_f)_{f \in \FF}$  be a  family of independent, not necessarily
identically distributed
random variables taking values in $\T$. We assume that the distribution of $\omega_f$
is absolutely continuous and we denote 
by $v_f$ its density function defined on $\T$.
We will denote the probability  space by $\Omega:=\T^\FF$
and the expectation w.r.t. this probability measure by $\E$.
By $A_\omega$ we shall denote a vector potential such that $\rd A_\omega = B_\omega$.

The Wegner estimate  will hold under the  assumption that $v_f$ is a twice continuously
differentiable function with uniformly bounded second derivative. Moreover,
we will assume that $v_f$ is supported away from integer multiples of $\pi$.
Note that not only the ``minimal'' flux $\om_f\approx 0$ is excluded, but
also the ``maximal'' flux $\om_f\approx \pi$.
To formulate this precisely, for any $b\in (0, \pi/2)$ we introduce 
the following subset of the torus 
\be
\T_{b} := \T \setminus \big\{  (- b,b) \cup  (\pi - b, \pi + b  )  \big\} .
\label{def:T}
\ee

\bigskip

\noindent
{\bf Assumption A($\boldsymbol{b,D}$).}   
Let $v_f \in C^2(\T)$ with $\| v_f \|_{C^2} \leq D$ and ${\rm supp \,} v_f \subset \T_b$.

\bigskip

For a rectangle $\Lambda \subset \Z^2$, we consider 
the Hamiltonian  $H_{\Lambda}(A)$  restricted to $\ell^2(\Lambda)$ as 
follows: for $\psi \in \ell^2(\Lambda)$ and 
$x \in \Lambda$ we set 
$$
[H_\Lambda(A) \psi](x) := 4 \psi(x) - \sum_{\substack{ y \in \Lambda \\ |x-y|=1}} 
 e^{i A(x,y)} \psi(y)  .
$$
This choice of boundary conditions  is referred to as {\it simple boundary conditions}, 
see \cite{KL}. 
With respect to these boundary conditions we state the  Wegner estimate. We remark that
the proof of the localization presented in
 \cite{KL} uses this choice of boundary conditions for the  Wegner estimate and 
for the box Hamiltonians in the multiscale analysis.

We introduce   the cube centered at the origin of length $L\in\N$ by   
$$
\Lambda_L := \{ x \in \Z^2 : {\rm max}\{|x_1|,|x_2| \}  \leq L \} . 
$$
We shall write $H_L = H_{\Lambda_L}$.  By $\chi_{E,\eta}$  we will denote 
the characteristic function of the closed interval  $[E - \eta/2, E + \eta/2]$.  
 The Wegner estimate holds in an energy interval up to 
 $$E_{\rm crit} := 4 - \sqrt{8} = 1.1715...$$
which is the the maximal possible  value for the bottom  of the spectrum, see \eqref{boundonspec}.
We now state our main result on the Wegner estimate.

\begin{theorem} \label{thm:weg} Suppose Assumption {\bf A($\boldsymbol{b,D}$)} holds for 
some $b\in (0,\pi/2)$ and a finite constant $D$. Then for any $E^* < E_{\rm crit}$
 there exists a finite  constant $C =C(b,D,E^*)$  such that 
for any $E \geq 0$
and $\eta \geq 0$ with  $E + \eta /2 \leq E^*$, we have 
$$
{\bf E}\, {\rm Tr} \chi_{E,\eta}(H_L(A_\omega)) \leq C \eta  L^8 . 
$$
\end{theorem}

\bigskip
\noindent
{\bf Remark.}  By   using the unitary transformation $\psi(x) \to (-1)^{x_1 + x_2}\psi(x)$,
we have the unitary equivalence
$H(A) \cong 8 - H(A)$.
Therefore the
 Wegner estimate holds likewise for $E - \eta/2 \geq  8 - E^*$ and the same symmetry
applies for the result about  localization (Theorem~\ref{thm:loc}  below).
The restriction $E^* < E_{\rm crit}$ 
originates from the   proof of   Lemma~\ref{lemma:lowerreg}.
Since the Lifshitz asymptotics \eqref{eq:lifshitz} only holds for small 
energies,  a Wegner estimate for energies below $E_{\rm crit}$  is sufficient  for 
the Anderson localization near the band edge which is our main application of Theorem~\ref{thm:weg}.
Nevertheless it is an interesting question on its own to establish
the largest possible upper threshold $E^*$ of the range of energies $E$ for which
 one could  extend   Lemma~\ref{lemma:lowerreg} and hence the Wegner estimate.

\bigskip

We now state the result on localization. For this part, 
 we assume that  the random fluxes $\om_f$ are not only 
independent but also identically distributed with density $v$.
Under this condition,
 $H(A_\om)$ is an ergodic operator, and its
spectrum $\sigma(H(A_\om)) =\Sigma$ is actually 
independent of $\omega$ a.s., see \cite{CL} (Proposition V.2.4.).  

For $E \in \R$, the integrated density of states is defined by 
$$
k(E) = \lim_{L \to \infty} \frac{1}{|\Lambda_L|} \# \big\{ {\rm  eigenvalues \  of \ } H_{\Lambda_L}(A_\omega) \leq E \big\} ,
$$
where the limit exists $\omega$ a.s. and is independent of the choice of sample $\omega$, see
Appendix C of \cite{Na1}. The density of states is independent 
of the choice of boundary conditions, which can be seen using the min-max principle \cite{Na1}.
Lifshitz asymptotics is shown in \cite{KNNN} under the 
assumptions that  
${\rm supp}\, v \subset \T \setminus (-c,c)$, $\pm c \in {\rm supp}\, v$ for some $0<c<\pi$
 and  $v$ is Lipshitz continuous on $\T \setminus (-c,c)$. 
Under these assumptions 
\be \label{boundonspec}
\Sigma = [ E_0, 8 - E_0 ] ,\quad {\rm with} \quad E_0= E_0(c):= 4(1- \cos(c/4)) ,
\ee
\cite{Na1,KNNN}.
In \cite{KNNN} it is shown (Theorem 1.2)  that 
\be \label{eq:lifshitz}
\limsup_{E \downarrow E_0} \frac{\log(-\log(k(E))}{\log(E-E_0)} \leq - 1 . 
\ee
This result implies, roughly speaking, 
$$
k(E) \lesssim e^{-(E-E_0)^{-1+\delta}} \quad {\rm as} \ \ E \downarrow E_0 .
$$
for any $\delta>0$. 
 Lifshitz tail estimate and    Theorem \ref{thm:weg} imply localization at the bottom 
of the spectrum using standard arguments, see \cite{CH,S,KL} for details.
Here we only state the result noting that the assumptions ensure that $E_0 < E_{\rm crit}$ 
and thus a Wegner estimate always holds at the bottom of the spectrum.

\begin{theorem}\label{thm:loc} Suppose Assumption {\bf A($\boldsymbol{b,D}$)} holds for 
some $b\in(0,\pi/2)$ and a finite constant $D$. Assume that the random fluxes 
$\om_f$
are independent, identically distributed and  $\pm b \in {\rm supp} v$.
Then Anderson localization
holds near the bottom of the spectrum. Namely, there exists an $E_{\rm loc} > E_0(b)$ such
that $H(A_\omega)$ has dense pure point spectrum on $[E_0(b), E_{\rm loc}]$ almost surely,
and each eigenfunction associated to an energy in this interval decays exponentially. 
\end{theorem}

\section{Proof of the Wegner estimate: Theorem \ref{thm:weg}}

Let  $\alpha_f$,  denote a vector potential 
 of a magnetic field of flux 1 through the square $f$.
I.e., $\rd\alpha_f = \delta_f$ where $\delta_f$ is the function on $\FF$ which is one 
at $f$ and zero otherwise.
Later in Section~\ref{sec:prop}
we will choose specific  gauges, see \eqref{eq:gauge} below.  
In the current section we only need that the absolute value of $\alpha_f$ is bounded by one.

Let $\Lambda = \Lambda_L$ for some $L \in \N$.
Define  $\FF_\Lambda :=  \{ f \in \FF : f \subset \Lambda \}$ and 
$\AR_\Lambda := \{ a \in \AR : a \in \Lambda\times \Lambda \}$. 
Henceforth we restrict the functions $\alpha_f$ to 
$\AR_\Lambda$.
We define  the vector potential
$$
A_\omega =  \sum_{f \in \FF_\Lambda} \omega_f \alpha_f ,
$$
then $\rd A_\omega = B_\omega$. 
 Moreover
we will drop the subscript $\omega$ in the notation and set  $H = H_L(A)$
in this section.

We will use perturbation theory, we abbreviate $Y_f= \frac{\pt}{\pt \om_f}$,
and for any state $\psi$ we introduce the expectation of the
$Y_f$-derivative of the energy in state $\psi$:
$$
 \langle Y_f H \rangle_\psi:=
 - \sum_{\substack{(x,y)\in \Lambda^2 \\  |x-y|=1}}
 \overline{\psi(x)}  i  \alpha_f^{}(x,y) e^{i A^{}(x,y)}  \psi(y).
$$
Let $\lambda$ be a non-degenerate eigenvalue of $H$ with a normalized eigenfunction $\psi$. 
In that case the eigenvalue $\lambda$ is a function of 
the random variables $\{ \omega_f \}$. For each $f \in \FF$, we have 
by the Hellmann-Feynman theorem and a straight forward calculation  that
\begin{align} \label{ylambda1}
Y_f \lambda= \frac{\partial \lambda }{\partial {\omega_f}}  = \langle Y_f H \rangle_\psi.
\end{align}

To give a physical meaning to the right hand side
 we introduce the current of a wavefunction in the Hilbert space. 
 For  $\varphi \in \ell^2(\Lambda)$, we 
define the  current $J_\varphi$ as the  function on $\AR_\Lambda$ 
given by  
$$
J_\varphi(a) :=  - 2  {\rm Re} \overline{\varphi(a_i)} i e^{ i A(a)} \varphi(a_t) .
$$

The current does not depend on the chosen gauge and  observe that 
\be \label{eq:symofcurr}
J_\varphi(\overline{a}) = - J_\varphi(a) .
\ee
If $\varphi$ is
 an eigenvector then it is a straightforward calculation to show that 
the ``divergence'' of $J_\varphi$ vanishes, namely
for all $x \in \Lambda$
\be \label{divzero} 
 \sum_{\substack{e :  \ |e|=1 \\ x+e \in \Lambda}} J_\varphi(x,x+e) = 0 .
\ee

We can write  the right hand side of  \eqref{ylambda1}  as 
\be
     \langle Y_f H \rangle_\psi =   \frac{1}{2}  \sum_{a \in \AR_\Lambda} \alpha_f(a) J_\psi(a) .
\label{ylambda}
\ee
Notice, since we sum over directed edges we have a factor $\frac{1}{2}$. 
We also remark that the left hand side is gauge invariant,
while the right hand side seems to depend on the gauge $\alpha_f$.
The gauge independence of the right hand side follows from the
fact that the divergence of $J_\psi$ is zero \eqref{divzero}.

\bigskip

In Section~\ref{sec:pr} we will prove the following key technical estimate: 
\begin{proposition}\label{prop:lower}  \label{prop51}
Let $b\in (0,\pi/2)$ and $E^* < E_{\rm crit} $. 
Then there exists a constant $C=C(b,E^*)$ such that 
\be
   \sum_{f \in \FF_\Lambda}      \langle Y_f H \rangle_\psi^2 \ge C^{-1}L^{-4}
\label{squarelow}
\ee
for any normalized eigenunction $\psi$ with eigenvalue  in $(-\infty , E^* ]$
and for any collection of fluxes $\om_f\in \T_b$.
\end{proposition}

\bigskip

\noindent
{\bf Proof of Theorem \ref{thm:weg}}.
Let $\lambda_1,\lambda_2,...$ denote  the eigenvalues of $H$. First we assume that they are simple.
We fix $E^* < E_{\rm crit}$.
In the sequel we set $\chi = \chi_{E,\eta}$, with $E + \eta/2 \leq E^*$.
{F}rom Proposition \ref{prop:lower} and \eqref{ylambda1}  it easily follows that
\be
   \Tr \chi(H) = \sum_\ell \chi(\lambda_\ell)
\leq CL^4
  \sum_\ell \sum_{f \in \FF_\Lambda} (Y_f\lambda_\ell)^2 \chi(\lambda_\ell) ,
\label{main1}
\ee
To estimate the right hand side we 
 introduce the following functions,
$$
F(x) := \int_{-\infty}^x \chi(t) \rd t , \quad G(y) := \int_{-\infty}^y F(x) \rd x .
$$
For any $Y=Y_f$ and any $\ell$ we have by the chain rule and  Leibniz
\be\label{main2}
(Y\lambda_\ell)^2
\chi(\lambda_\ell) = Y^2 G(\lambda_\ell) -   (Y^2 \lambda_\ell)F(\lambda_\ell) .
\ee
To estimate the sum of the second term over all eigenvalues $\lambda_\ell$, we will use 
 the following Lemma, for a proof  see \cite{EH2}.
\begin{lemma}\label{lm:trick}
We have for any $Y=Y_f$
\be
    \Tr (Y^2 H) F(H)\le  \sum_\ell(Y^2\lambda_\ell)F(\lambda_\ell) .
\label{trick}
\ee
\end{lemma}
Thus combining \eqref{main1}, \eqref{main2} and \eqref{trick}, we have
\be
   \Tr \chi(H)\leq CL^4
\sum_{f \in \FF_\Lambda}\Big( \Tr Y_f^2 G(H) - \Tr (Y_f^2 H)F(H)\Big).
\label{com}
\ee
To estimate the second term we use that $\| F \|_\infty \leq C \eta$. Moreover, we use 
that  $\| Y_f^2 H \| \leq 4$, which   can be seen using that for any  $\varphi \in \ell^2(\Lambda)$
we have 
$$
(Y_f^2 H \varphi)(x) =  - \sum_{|x-y|=1}  \alpha_f^2(x,y) e^{i A(x,y)} \varphi(y) .
$$
To estimate the the first term in \eqref{com} we integrate by parts after taking expectation,
$$
    \bE \Tr Y_f^2 G(H) = \int \prod_{\zeta\in \FF_\Lambda}
  v_\zeta(\om_\zeta)\rd \om_\zeta
\frac{\partial^2}{\partial \om_f^2}  \Tr G(H)
  = \int  \prod_{\zeta\neq f} v_\zeta(\om_\zeta)\rd \om_\zeta
  \int v_f''(\om_f) \rd \om_f \Tr G(H).
$$
Using that $G(y)\leq C\eta y$ we can estimate  $| \Tr G(H)| \leq C L^2 \eta $.
Thus we obtain 
\be \label{eqofwegner}
 {\bf E} \Tr \chi(H)\leq  C  \eta L^{8} ,
\ee
and the  Theorem  follows in the non-degenerate case. 

\bigskip

In the argument above we used that $\lambda=\lambda_\ell$ is
a simple eigenvalue since we implicitly  assumed that $\lambda_\ell$
is differentiable w.r.t. $\om_f$. Differentiability of the eigenvalues
may be ensured even in case of degeneracy  by choosing 
an appropriate branch, but this choice  may depend on $f$ hence
a more careful treatment is needed which we explain now.

Suppose that $\lambda$ is  a degenerate eigenvalue 
 with degeneracy $m_\lambda$ for some fixed $\om$. Let 
 $P_\lambda$ denote the eigenprojection and we fix 
$\{ \varphi_{\lambda, h} \}_{h=1,...,m_\lambda}$ an orthonormal basis of 
the corresponding eigenspace. Moreover, 
for each fixed $f \in \FF_\Lambda$  
we can also choose orthonormal eigenfunctions 
$\psi_{f,\lambda,h}$, with  $h=1,...,m_\lambda$, which 
are real analytic functions of $\omega_f$ (if the other variables of $\omega$ are kept fixed)
and their eigenvalues 
$E_{f,\lambda,h}$ 
are real analytic in $\omega_f$ as well, see \cite{K} (Theorem 1.10. Section II).
By analytic perturbation theory, the derivatives 
$Y_f E_{f,\lambda, h}$ ($h=1,2, \ldots, m_\lambda$)
are the eigenvalues of the matrix
$$
   T := P_E (Y_f H) P_E  \upharpoonright {\rm ran} P_E.
$$
 By   Jensen's inequality,   we  have 
\be \label{basicdegen}
\sum_{h=1}^{m_\lambda} (Y_f E_{f,\lambda, h})^2  = {\rm Tr} T^2 =  
\sum_{h=1}^{m_\lambda} ( \varphi_{\lambda, h} , T^2 \varphi_{\lambda,h} ) 
\geq \sum_{h=1}^{m_\lambda} ( \varphi_{\lambda,h} , T  \varphi_{\lambda,h})^2 
 = \sum_{h=1}^{m_\lambda} \langle Y_f H \rangle_{\varphi_{\lambda,h}}^2 .
\ee
Summing \eqref{basicdegen} over $f$, we find, using  Proposition \ref{prop:lower}, that for $\lambda \leq E^*$,
\be \label{eq:prop31plus}
\sum_{f \in \FF_\Lambda} \sum_{h=1}^{m_\lambda} (Y_f  E_{f,\lambda,h })^2 \geq m_\lambda C^{-1} L^{-4} .
\ee
We can now proceed as before. Using \eqref{eq:prop31plus} and  \eqref{main2}, we find 
\begin{align*}
{\rm Tr} \chi(H) &= \sum_{\lambda \in \sigma(H)} m_\lambda \chi(\lambda) \\
&\leq \sum_{\lambda \in \sigma(H)} C L^4 \sum_{f \in \FF_\Lambda} \sum_{h=1}^{m_\lambda} (Y_f E_{f,\lambda, h } )^2 
\chi( E_{f,\lambda,h}   )    \\
&=  C L^4 \sum_{f \in \FF_\Lambda} 
 \sum_{\lambda \in \sigma(H)}  \sum_{h=1}^{m_\lambda} \left(  Y_f^2 G(E_{f,\lambda,h}) - 
( Y^2_f E_{f,\lambda,h} ) F(E_{f,\lambda, h}) \right)\\
&\leq  CL^4
\sum_{f \in \FF_\Lambda}\Big( \Tr Y_f^2 G(H) - \Tr (Y_f^2 H)F(H)\Big) ,
\end{align*}
where we used
  Lemma \ref{lm:trick} in the last step. 
 The rest of the proof is the same as in the case of non-degenerate eigenvalues
and we thus proved  Theorem \ref{thm:weg}. \qed

\bigskip

We remark that the possible degeneracy of the eigenvalues  may 
also be treated by standard perturbation theory for almost all $\omega$, using the 
fact that the interior of the set of $\omega$'s where there is no eigenvalue 
crossing has full measure.

\section{Proof of Proposition \ref{prop51}}\label{sec:prop}

\label{sec:pr}

We recall that by  \eqref{ylambda} we have 
$$
 \langle Y_f H \rangle_\psi =   \frac{1}{2} \sum_{a \in \AR_\Lambda} \alpha_f(a) J_\psi(a) .
$$
 The following lemma inverts this linear relation and
expresses the current in terms of $ \langle Y_f H \rangle_\psi  $.

\begin{lemma} \label{currentasderiv}
For  $a \in \AR$, let $f_a$ be the unique square in $\FF$ such that $a \in \partial f$.
Then for $a \in \AR_\Lambda$, we have 
\be \label{currentasderivative}
J_\psi(a) =  c_a     \langle Y_{f_a} H \rangle_\psi      -   
c_{\overline{a}}   \langle Y_{f_{\overline{a}}}  H \rangle_\psi    ,
\ee 
where $c_a=1$ if $f_a \in \FF_\Lambda$ and $c_a = 0$ otherwise. 
\end{lemma}

\bigskip
\noindent
{\bf Proof of Lemma \ref{currentasderiv}}. We will introduce four different gauges, $\al^{(\tau)}$,
$\tau=1,2,3,4$, for the magnetic field $\delta_f$ whose flux is 1 through the square $f$
and zero elsewhere.
 For  $f=f_x \in \FF$ with $x=(x_1,x_2)$ we set 
\begin{align}  \label{eq:gauge}
\alpha_f^{(1)}(y,y + {\bf e}_2 )  &= 0 \; , \qquad   \alpha_f^{(1)}(y,y + {\bf e}_1 ) = \left\{ \begin{array}{ll} -1  \ , & {\rm if} \ 
 y_2 > x_2, \, y_1 = x_1 \\ 0 \ , & {\rm otherwise } \end{array} \right. , \\
\alpha_f^{(2)}(y,y + {\bf e}_1 )  &= 0 \; , \qquad 
\alpha_f^{(2)}(y,y + {\bf e}_2 ) = \left\{ \begin{array}{ll} 1 \ , & {\rm if} \ y_1 > x_1, \, 
y_2 = x_2 \\ 0 \ , & {\rm otherwise } \end{array} \right.   , \nonumber \\
\alpha_f^{(3)}(y,y + {\bf e}_2 )  &= 0   \; ,
 \qquad   \alpha_f^{(3)}(y,y + {\bf e}_1 ) = \left\{ \begin{array}{ll} 1  \ , & {\rm if} \ 
 y_2 \leq  x_2, \, y_1 = x_1 \\ 0 \ , & {\rm otherwise } \end{array} \right. ,  \nonumber \\
\alpha_f^{(4)}(y,y + {\bf e}_1 )  &= 0 \; , \qquad 
\alpha_f^{(4)}(y,y + {\bf e}_2 ) = \left\{ \begin{array}{ll} - 1 \ , & {\rm if} \ y_1 \leq  x_1, \, y_2 = x_2 \\ 0 \ , & {\rm otherwise } \end{array} \right.   , \nonumber 
\end{align}
and extend the definition  to $\AR$ by \eqref{defofG}.

Now we can prove Lemma \ref{currentasderiv}.
Let  $a = (y,y+{\bf e}_1) \in \AR_\Lambda$ be a horizontal edge.  
For  $-L < y_2$, the
 gauge $\tau=1$ gives \eqref{currentasderivative}. For the boundary case $y_2=-L$
we can use the gauge $\tau=3$. The identity for the edge  $\overline{a}$ follows using  \eqref{eq:symofcurr}.
The vertical edges follow similarly using the gauge $\tau=2,4$.
\qed

\bigskip

The next lemma  gives a lower bound on the current of an eigenfunction.
The proof  will be given in Section \ref{sec:proofofreg}.

\begin{lemma} \label{lemma:lowerreg} Let $b\in(0,\pi/2)$ and $E^* < E_{\rm crit}$. 
Then there 
exists a constant $C = C(b,E^*)$ such that 
$$
\sum_{{ a \in \AR_\Lambda }} |J_\psi(a)|^2 \geq C L^{-4} 
$$
for all normalized eigenfunctions $\psi$ of $H_\Lambda(A)$, with energy $E \leq E^*$,
and for any  vector potential $A$ with $\rd A \in \T_{b}$ (see \eqref{def:T}).
\end{lemma}

\medskip

Note that the flux is required to be separated away from 0  and $\pi$,
i.e. not only the zero flux is excluded but also the
``maximal'' flux. The first exclusion is obvious since if $\rd A=0$, then the eigenfunction can be chosen
real and then clearly $J_\psi=0$. If $\rd A=\pi$ on each plaquet,
then one can choose a gauge $A$ with $A(x,x+{\bf e}_1)=0$ and $A(x,x+{\bf e}_2)= \pi x_1$
and extend it by \eqref{defofG}. In that case the Hamiltonian $H_\Lambda(A)$ is real.
In particular, the eigenfunctions are real for the maximal flux, hence
the current vanishes in this case as well.

\bigskip
\noindent
{\bf  Proof of Proposition \ref{prop51}.} Using Lemma  \ref{currentasderiv}
  we have  
$$
|J_\psi(a)|^2 \leq  2  c_a    \langle Y_{f_a} H \rangle_\psi^2      +   2  c_{\overline{a}}
\langle Y_{f_{\overline{a}}} H \rangle_\psi^2
$$
Observing that each square gives rise to 8 directed edges, after summing up
this inequality for all $a\in \AR_\Lambda$, we find  
$$
32 \sum_{f \in \FF_\Lambda} \langle Y_f H \rangle_\psi^2 \geq \sum_{a \in \AR_\Lambda} |J_\psi(a)|^2 .
$$
The proposition now 
follows from  Lemma \ref{lemma:lowerreg}.
\qed

\section{Proof of the regularity lemma: Lemma  \ref{lemma:lowerreg}}
\label{sec:proofofreg}

Let $\psi \in \ell^2(\Lambda)$ be an eigenfunction of $H_\Lambda(A)$ with 
energy $E$. If $\psi(x)\ne 0$,
we write 
\be \label{eq:polar}
\psi(x) = e^{i \lambda(x)} |\psi(x)|
\ee
for some real function $\lambda(x) \in \T$. Then the current is 
\be \label{eq:polarcur}
J_\psi(a) =   2|\psi(a_i)| |\psi(a_t)| \sin(\varphi_a),
\ee
with $\varphi_a := A(a) + \lambda(a_t) - \lambda(a_i)$. 
The goal will be to find a unit square $Q \in \FF_{\Lambda}$ such that
  $\inf_{x \in Q} |\psi(x)|$ 
is bounded from below.  Using that the magnetic flux through $Q$
is separated away from zero and $\pi$,
one can  then show that 
there is  a nonvanishing current along the boundary of that square.

To this end let $x_0 \in \Lambda$ be a point where $\psi$ attains its maximum absolute value,
\be
|\psi(x_0)| = M := {\rm sup}_{x \in \Lambda} | \psi(x) | .
\label{psimax}
\ee
In this section we use the convention that $\psi(x) =0$ if $x \notin \Lambda$.

\begin{lemma} \label{lem:decay1} Let the energy $E$ of the eigenfunction $\psi$ 
satisfy
$0 \leq E \leq 4$. Then the  following statements hold.

\begin{itemize}
\item[(a)] We have 
$$
\max_{\substack{ y \in \Lambda \\  |y-x_0|=1}} |\psi(y)| \geq (1- E/4) M.
 $$
\item[(b)] Suppose $\tilde{y}_0 \in \Lambda$
is a  nearest neighbor of $x_0$  with   $|\psi(\tilde{y}_0)| \leq \epsilon M $
for some $\epsilon\ge 0$.
Then
$$
\min_{\substack{ y \in \Lambda \setminus \{ \tilde{y}_0 \} \\ |y-x_0|=1 }} | \psi(y) | \geq 
(2 - E - \epsilon ) M .
$$

\item[(c)]  Suppose   $y_0 \in \Lambda$
is a  nearest neighbor of $x_0$  with   $|\psi(y_0)| \geq  \kappa M $ for some $\kappa\ge0$.
Then
$$
\left\{ (4 -  E )\kappa -2 \right\}  M  \leq 
| \psi(  y_0 + {\bf t} )| + | \psi(  y_0 - {\bf t} )| ,
$$   
where ${\bf t}$ is a unit vector orthogonal to  $y_0  -x_0$.
\end{itemize}
\end{lemma}
\begin{proof} 
 Adjusting the phase of the 
eigenvector we can assume that $\psi(x_0) = M$. 
Since the statements of the lemma are gauge invariant,
for notational simplicity we may choose for (a) and (b) a gauge, such that 
$A(x_0,y)=0$ for all $y$ with  $|x_0-y|=1$. 
The statement (a) follows from  the eigenvalue equation 
$$
\sum_{\substack{ y \in \Lambda \\ |y-x_0|=1}} \psi(y) = (4 - E ) \psi(x_0) .
$$
and by
taking real parts
$$
4 \max_{\substack{y \in \Lambda \\ |y-x_0|=1 }} {\rm Re} \psi(y) \geq ( 4 - E) \psi(x_0).  
$$
For part   (b), we
 fix  $y \in \Lambda \setminus \{ \tilde{y}_0 \} $ such that $|y - x_0 |=1$. 
 Then  by the eigenvalue equation 
\be \label{eq:eigen}
\psi(y) = ( 2  - E) \psi(x_0)   - \psi( \tilde{y}_0 )
 +    2 \psi(x_0) - \sum_{\substack{z \in \Lambda \setminus \{\tilde{y}_0 , y \} \\ |x_0-z|=1 }}
   \psi(z) .
\ee
Now taking the real part of both sides we find using that  ${\rm Re} [\psi(x_0)-\psi(z)] \geq 0$, 
$$
{\rm Re}\, \psi(y) \geq (2 - E) \psi(x_0)  - \epsilon M .
$$
This yields (b). 
To prove (c), we choose a gauge such that $A(y_0,y)=0$ for all $|y_0 - y|=1$.
By the eigenvalue equation at $y_0$
$$
(4 - E) \psi( y_0) = \psi(x_0) + \psi(  2 y_0  -x_0 ) + 
\psi(  y_0 + {\bf t}) +  \psi( y_0 - {\bf t}   ) .
$$
Hence by the triangle inequality 
$$
(4 - E ) | \psi( y_0)|  \leq  2 M + | 
\psi(  y_0 - {\bf t} )|  + | \psi(  y_0 + {\bf t} )| ,
$$
and (c)  now follows. 
\end{proof}

\bigskip

With the above  lemma we can show  the following proposition. 

\begin{proposition} \label{propsquare}   Let $E^* <   E_{\rm crit}= 4 - \sqrt{8} $.
 There exists a positive number $c = c(E^*)$ such that for any 
eigenfunction $\psi$ with eigenvalue $E \leq E^*$ there 
exists a cube $Q \in \FF_\Lambda$  such that 
\be \label{propineqQ}
{\rm min}_{x \in Q} |\psi(x)| \geq  c M .
\ee
\end{proposition}
\begin{proof} Let $\epsilon =1/10$ and recall the definition
of $M$ and $x_0$ from \eqref{psimax}. We consider the following two cases.

\bigskip

\noindent
\underline{Case 1:} ${\rm min}_{y \in \Lambda : |y -x_0| = 1} |\psi(y)| >  \epsilon M $.

\bigskip

By  Lemma \ref{lem:decay1} (a)  there exists a nearest neighbor 
$y_0 \in \Lambda$ of $x_0$ such that
$$
 |\psi(y_0)| \geq (1- E^* /4) M >0.
$$
Then by  Lemma \ref{lem:decay1} (c) (and using the notation introduced there) 
there exists a $\sigma \in \{-1,1\}$ such that 
$$
|\psi(y_0 + \sigma {\bf t})| 
 \geq \frac{1}{2} \left\{ (4 - E^*) (1 - E^*/4) -2 \right\} M .
$$
Thus choosing the square 
$Q =\{ x_0, x_0 + \sigma {\bf t}, y_0, y_0 + \sigma {\bf t}\}$,
the estimate \eqref{propineqQ} follows from the choice of $E^*< 4-\sqrt{8}$. 

\bigskip

\noindent
\underline{Case 2:} ${\rm min}_{y \in \Lambda : |y -x_0| = 1} |\psi(y)| \leq   \epsilon M$.

\bigskip

Let $\tilde{y}_0 \in \Lambda$ be the nearest neighbor of $x_0$ with  
$|\psi(\tilde{y}_0) | =  {\rm min}_{y \in \Lambda : |y -x_0| = 1} |\psi(y)|$. Let $y_0 = 2 x_0 - \tilde{y}_0$. Then by  Lemma \ref{lem:decay1} (b) and (c)
there exists a $\sigma \in \{-1,1\}$ such that 
$$
|\psi(y_0 + \sigma {\bf t}) |
 \geq \frac{1}{2} \left\{ (4 - E^*) (2 - E^* - \epsilon ) -2 \right\} M .
$$
Thus choosing the square 
$Q =\{ x_0, x_0 + \sigma {\bf t}, y_0, y_0 + \sigma {\bf t}\}$,
the estimate \eqref{propineqQ} 
follows from  Lemma \ref{lem:decay1} (b) and the choice of $E^*$ and $\epsilon$.
\end{proof}

\bigskip

With the above Proposition  we can show   Lemma \ref{lemma:lowerreg} using 
that a nonzero flux produces a current.

\bigskip

\noindent
{\bf Proof of Lemma \ref{lemma:lowerreg}}.
  Since $\psi$ is $\ell^2$-normalized,
clearly, $|\psi(x_0)| \geq 1/L$, where $x_0$ is defined by \eqref{psimax}.  
 By Proposition   \ref{propsquare}
there exists  a unit square $Q \in \FF_\Lambda$  such
 that for a constant $c$ (depending on $E^*$), we have 
$$
\inf_{x,y\in Q} |\psi(x) \psi(y)| \geq c L^{-2} .
$$
With respect to   the parameterization  \eqref{eq:polar} we have by   \eqref{eq:polarcur}
$$
\sum_{a \in \partial Q} | J_\psi(a)|^2 \geq c^2 L^{-4} \sum_{a \in \partial Q} \sin^2(\varphi_a) .
$$
This 
will imply the lemma using  that for some positive  constant $c$  (depending on $b$)
\be
 \sum_{a \in \partial Q} \sin^2(\varphi_a)  \geq c   \label{eq:lowboundonmag} .
\ee
To show this, we first note that 
$ \omega_Q = \sum_{a \in \partial Q} \varphi_a  \in  \T_{b}$. 
We will show that this implies that 
there is an  $a \in \partial Q$ such that 
$\varphi_a \in \T_{b/8}$. The estimate \eqref{eq:lowboundonmag} will then follow since
 $\sin^2(\cdot)$ is bounded from below on $\T_{b/8}$ by a positive constant (depending
on $b$). 
Suppose that there would not be an  $a \in \partial Q$ such that 
$\varphi_a \in \T_{b/8}$. Then $\varphi_a = n_a \pi + b_a$ with $n_a \in  \Z$ and 
$|b_a| \leq  b/8$.  In that case $\omega_Q = n \pi + \tilde{b}$,
with $n \in \Z$ and $|\tilde{b}| \leq b/2$. This implies
$\omega_Q \notin \T_b$, which is a contradiction.
\qed

\end{document}